\documentclass[12pt]{article}
\usepackage{amsmath}
\usepackage{amsfonts}
\usepackage{amssymb}
\usepackage{amsthm}
\usepackage{graphicx}
\usepackage{fullpage}
\usepackage{setspace}
\usepackage{verbatim}
\usepackage{natbib}
\usepackage{appendix}
\usepackage{rotating}

\newtheorem{theorem}{Theorem}
\newtheorem{lemma}{Lemma}

\title{On the Asymptotic Distribution of Variance Weighted KS
  Statistics}

\author{Timothy B. Armstrong\thanks{email: timothya@stanford.edu.
    This paper was written with generous support from a
    fellowship from the endowment in memory of
    B.F. Haley and E.S. Shaw through the Stanford Institute for Economic
    Policy Research.}
\\
Stanford University}

\begin{document}

\maketitle

\begin{abstract}
This paper derives the asymptotic distribution of variance weighted Kolmogorov-Smirnov statistics for conditional moment inequality models for the case of a one dimensional covariate.  The asymptotic distribution depends on the data generating process only through the variance of a single random variable, leading to critical values that can be calculated analytically.  By arguments in \citet{armstrong_weighted_2011}, the resulting tests achieve the best minimax rate for local alternatives out of available approaches in a broad class of settings.
\end{abstract}

\section{Introduction}\label{introduction_sec}

This paper derives the asymptotic distribution of variance weighted Kolmogorov-Smirnov (KS) statistics for conditional moment inequality models for the case of a one dimensional conditioning variable.  The arguments can be extended to higher dimensional conditioning variables as well.  The asymptotic distribution is extreme value, with a scaling that can be easily estimated.  Thus, critical values for this test are easy to compute and do not require resampling or simulation from complicated random processes.  By arguments in \citet{armstrong_weighted_2011}, these tests are, in a broad class of models, the only tests available that achieve the best minimax rate for power against local alternatives out
 of available tests.  Loosely speaking, the test will be optimal in this sense generically when the conditioning variable is continuously distributed and the model may not be point identified.  In cases where point identification is suspected, other approaches may be desirable, although the results in \citet{armstrong_weighted_2011} show that the tests in the present paper will still have good power properties.

The tests proposed in this paper are most closely related to those studied by \citet{armstrong_weighted_2011} and \citet{chetverikov_adaptive_2012}.  \citet{armstrong_weighted_2011} uses a similar class of statistics, but proposes confidence regions that satisfy the stronger criterion of containing the entire identified set.  That paper derives critical values that are conservative even for this stronger criterion.  While the critical values proposed in that paper would be valid in this context, they rely on conservative bounds for the null distribution of the test statistic.  The asymptotic distribution result in the present paper allows for less conservative critical values, leading to more powerful tests.  \citet{armstrong_weighted_2011} shows that, even with these conservative critical values, the confidence regions proposed in that paper achieve the best minimax rate of convergence in the Hausdorff metric to the identified set out of available methods in a broad class of models.  These arguments will give analogous local power results for the tests proposed in the present paper.
\citet{chetverikov_adaptive_2012}, which was written at around the same time as the present paper, proposes critical values for the tests treated in the present paper using simulation from an approximating random process.  That paper uses a novel argument that proves validity of simulation methods without deriving an asymptotic distribution or even showing that one exists.  The asymptotic distribution result of the present paper shows that such arguments are not necessary, at least in the one dimensional case, and provides critical values that can be computed easily without simulation.
Asymptotic distribution results are also important for asessing the asymptotic behavior of the critical values for use in power calculations.
On the other hand, the methods of \citet{chetverikov_adaptive_2012} apply to higher dimensional conditioning variables.  While the results of the present paper could be extended to this case, the approach of \citet{chetverikov_adaptive_2012} allows computation of valid critical values even without this extension.  It also seems likely that, while being more difficult to compute, the simulation based critical values of \citet{chetverikov_adaptive_2012} enjoy higher order accuracy properties similar to resampling methods in other settings, and that the methods of that paper would be useful in showing this.

Other approaches to inference on conditional moment inequalities include
\citet{andrews_inference_2009},
\citet{kim_kyoo_il_set_2008}, \citet{khan_inference_2009},
\citet{chernozhukov_intersection_2009},
\citet{lee_testing_2011},
\citet{ponomareva_inference_2010},
\citet{menzel_estimation_2008} and \citet{armstrong_asymptotically_2011}.  See \citet{armstrong_weighted_2011} for a discussion of some of these approaches, including power results that show that the approach in the present paper is optimal among these approaches in a certain sense in a broad class of models in the set identified case.
The literature on the related problem of inference with finitely many unconditional moment inequalities is also recent, but more developed.  Articles include
\citet{andrews_confidence_2004}, %
\citet{andrews_inference_2008}, %
\citet{andrews_validity_2009}, %
\citet{andrews_inference_2010}, %
\citet{chernozhukov_estimation_2007}, %
\citet{romano_inference_2010}, %
\citet{romano_inference_2008}, %
\citet{bugni_bootstrap_2010},
\citet{beresteanu_asymptotic_2008},
\citet{moon_bayesian_2009},
\citet{imbens_confidence_2004}
and \citet{stoye_more_2009}.

\section{Setup}

We wish to test the null hypothesis
\begin{align}\label{ineq_def_eq}
E(Y_i|X_i)\ge 0 \,\,\, a.s.
\end{align}
using iid observations $(X_1,Y_1),\ldots (X_n,Y_n)$.  Typically, $Y_i=m(W_i,\theta)$ for some parametric function $m$ of data $W_i$ that may include $X_i$.
See \citet{armstrong_weighted_2011} for several examples of such models along with references to the literature.
The moment inequality defines the identified set $\Theta_0$ of parameter values that satisfy this inequality, and it is of interest to design tests that are powerful against alternative hypotheses within the model defined by parameter values $\theta$ such that $m(W_i,\theta)$ does not satisfy the moment restriction (\ref{ineq_def_eq}).  Arguments in \citet{armstrong_weighted_2011} show that the tests in the present paper are most powerful among existing approaches against such alternatives in a local minimax sense in a broad class of models used in empirical economics.

This paper considers inference using the test statistic
\begin{align}\label{ks_stat_trunc_set_eq}
T_n=\left|\inf_{\{s,t|\hat\sigma(s,t)\ge \underline\sigma_n\}}
   \frac{E_nY_iI(s<X<s+t)}{\hat\sigma(s,t)}\right|_{-}.
\end{align}
The results in the following section derive the asymptotic distribution of $T_n$.  The result gives easily computable critical values for this test statistic.

\section{Asymptotic Distribution}

The following theorem gives the asymptotic distribution of $T_n$.  The result is stated for a single moment inequality (that is, $Y_i$ is one dimensional).  This result can be used immediately for the case of higher dimensional $Y_i$ using Bonferroni bounds, which will be conservative, but will not result in a first order loss in power.  Alternatively, the results could also be extended to this case to get less conservative critical values.

\begin{theorem}\label{asym_dist_thm}
Suppose that $E(Y_i|X_i)\ge 0$ a.s., the data are iid and
\begin{itemize}
\item[(i)] $X_i$ is one dimensional, and, for some
  $\underline x<\overline x$,  $E(Y_{i}|X_i)=0$ only on
  $[\underline x,\overline x]$.
\item[(ii)] $X_i$ has a density bounded from above away from infintity and
  from below away from zero.
\item[(iii)] $\sigma_{Y}^2(x)\equiv var(Y_i|X_i=x)$ is bounded away from
  zero and infinity.
\item[(iv)] $Y_i$ is bounded by a nonrandom constant with probability one.
\item[(v)] $\underline \sigma_n n^{1/2}/(\log n)^{5/2}\to\infty$.
\end{itemize}
Define
$c_n=E(\sigma_{Y}^2(x)I(\underline x\le X_i\le \overline x))
  /\underline\sigma_n$ and
$\hat c_n=E_n(Y_{i}^2(x)I(\underline x\le X_i\le \overline x))
  /\underline\sigma_n$.
Then
\begin{align*}
\liminf_n
P\left(\sqrt{n}T_{n}\le (2\log c_n)^{1/2}
     + \frac{\frac{3}{2}\log\log c_n - \log (2\pi^{1/2})+r}
         {(2\log c_n)^{1/2}}
  \right)
\ge \exp(\exp(r)),
\end{align*}
and the same statement holds with $c_n$ replaced with $\hat c_n$.
\end{theorem}

It is interesting to note that the critical value for $T_{n}$ is, to a
first order approximation, given by
$(2\log (1/\underline\sigma_n^2)/n)^{1/2}$.  This is the same as the
first order approximation to the critical value if the infimum were taken
only over $\hat\sigma(s,t)=\underline\sigma_n$, which would essentially
be a kernel estimator with bandwidth proportional to $\sigma_n$.  Thus, up
to a first order power comparison, there is asymptotically no loss in
power in considering larger bandwidths.  If $\underline \sigma_n=K
n^{-\delta}$ for some $0<\delta<1/2$, the critical value will be
approximately equal to
$(n/(4\delta \log n))^{1/2}$, so, loosely, the critical value is
proportional to $\delta^{-1/2}$ times a scaling that does not depend on
$\underline \sigma_n$.

The asymptotic distribution derived in Theorem \ref{asym_dist_thm} will be attained when the conditional moment inequality binds on $[\underline x,\overline x]$.  If the conditional moment inequality is binding for some, but not all, values of $x$ in $[\underline x,\overline x]$, the theorem gives an upper bound for the limiting distribution, resulting in a potentially conservative critical value.  In practice, one can compute the critical values setting $[\underline x,\overline x]$ equal to the entire support of $X_i$, or by using a first stage estimate of this set.  Perhaps surprisingly, Theorem \ref{asym_dist_thm} shows that using the entire support of $X_i$ does not lead to an asymptotic increase in the critical value up to a first order comparison relative to using a smaller, nondegenerate set.

\section{Conclusion}

This paper derives the asymptotic distribution of variance weighted KS statistics for conditional moment inequality models.  The results apply to a one dimensional conditioning variable, but could be extended to higher dimensional conditioning variables.
The result gives easily computable critical values for these test statistics.  The tests give the best available minimax rate for local alternatives in a broad class of models used in empirical economics.

\section*{Appendix: Proof of Theorem \ref{asym_dist_thm}}

We derive the asymptotic distribution of (\ref{ks_stat_trunc_set_eq})
when the infimum is taken over a fixed set $\mathcal{X}$ and the
inequality binds on this set.
This paper treats the case where $X_i$ is one dimensional and continuously
distributed and $\mathcal{X}=[\underline x,\overline x]$, although the
latter assumption can be relaxed without changing the argument in any
essential way.
We also assume that $X_i$ has a density bounded from above away from
infinity and from below away from zero, although this can be relaxed by
transforming $X_i$.

First, we derive the asymptotic
distribution of
\begin{align}
\label{ks_stat_tilde_eq}
\tilde T_n
\equiv \left|\inf_{\{s,t|\hat\sigma(s,t)\ge \underline\sigma_n, (s,s+t)\in\mathcal{X}\}}
   \frac{E_n\tilde Y_{i}I(s<X<s+t)}{\tilde\sigma(s,t)}\right|_{-}.
\end{align}
where $\tilde Y_i=Y_i-E(Y_i|X_i)$
and $\tilde\sigma^2(s,t)=var(\tilde Y_iI(s<X_i<s+t))$.  Then, we show
that
\begin{align*}
\inf_{\tilde\sigma(s,t)\ge \underline\sigma_n}
  \tilde\sigma(s,t)-\hat\sigma(s,t)
\end{align*}
is bounded from below away from zero when scaled by an appropriate rate,
so that replacing $\tilde\sigma(s,t)$ by $\hat\sigma(s,t)$ will not
increase the statistic too much.

Define
\begin{align*}
\tilde X_i%
=g(X_i)\equiv \int_{\underline x}^{X_i} \sigma^2_Y(x)f_X(x) \, dx.
\end{align*}
Then
\begin{align*}
&E[\tilde Y_{i}^2I(0\le \tilde X_i\le t)]
=E[\tilde Y_{i}^2I(g^{-1}(0)\le X_i \le g^{-1}(t))]
=\int_{\underline x}^{g^{-1}(t)} \sigma^2_Y(x)f_X(x) \, dx  \\
&=g(g^{-1})(t)
=t.
\end{align*}
Since monotone transformations of $X_i$ do not change the statistic, we
can work with $\tilde X_i$.  The corresponding value of
$\tilde \sigma^2(s,t)$ will be
\begin{align*}
E(\tilde Y_{i}^2I(s<\tilde X_i<s+t))
=E(\tilde Y_{i}^2I(0<\tilde X_i<s+t))
-E(\tilde Y_{i}^2I(0<\tilde X_i<s))
=s+t-s=t.
\end{align*}
This transformation takes $\underline x$ to $0$ and $\overline x$ to
\begin{align*}
\int_{\underline x}^{\overline x} \sigma^2_Y(x)f_X(x) \, dx
=E(\tilde Y_{i}^2I(\underline x\le X_i\le \overline x))
\equiv \tilde x_{u}.
\end{align*}

The following lemma gives an approximation to the process by a Brownian motion using a Skorohod embedding.
\begin{lemma}
Let $A_n$ be any sequence of sets of $(s,t)$ such that
$t\ge \underline \sigma_n$ for all $n$ and $(s,t)\in A_n$ and let $a_n$ be
an arbitrary sequence going to infinity.
There exists a sequence of Brownian motions $\mathbb{B}_n(t)$ such that
\begin{align*}
&\inf_{(s,t)\in A_n} \frac{\sqrt{n}E_n\tilde Y_{i}I(s<\tilde X_i<s+t)}
  {\sqrt{t}}  \\
&\ge \inf_{(s,t)\in B_n}  [\mathbb{B}_n(s+t)-\mathbb{B}_n(s)]/\sqrt{t}
   \cdot (1+\mathcal{O}_P((\log n)^{1/4}/(n^{1/4}\underline \sigma^{1/2})))
\end{align*}
where $B_n$ is the set of $(s',t')$ such that $|s-s'|$ and $|t'-t|/t$
are both less than $a_n\cdot (\log n)^{1/2}/(n^{1/2}\underline \sigma_n)$
for some $(s,t)\in A_n$.

\end{lemma}
\begin{proof}
Let $\tilde X_{i:n}$ be the $i$th least value of $\tilde X_{i}$, and
let $\tilde Y_{i:n}$ be the corresponding value of $\tilde Y_{i}$.  By
a Skorohod embedding conditional on the $X_i$s, there is a Brownian motion
$\mathbb{B}_n(s)$ and stopping times $\tau_{1,n},\ldots,\tau_{n,n}$
such that
\begin{align*}
\frac{1}{\sqrt{n}} S_{n,k}
\equiv \frac{1}{\sqrt{n}}\sum_{i=1}^k \tilde Y_{i:n}
=\mathbb{B}_n(\tau_{k,n}/n).
\end{align*}

We have
\begin{align*}
&\frac{\sqrt{n}E_n\tilde Y_{i}I(s<\tilde X_i<s+t)}{\sqrt{t}}
=\frac{1}{\sqrt{n}} \frac{S_{n,k(s+t)}-S_{n,k(s)}}{\sqrt{t}}  \\
&=\frac{\mathbb{B}_n(\tau_{k(s+t),n}/n)-\mathbb{B}_n(\tau_{k(s),n}/n)}
  {\sqrt{\tau_{k(s+t),n}/n-\tau_{k(s),n}/n}}
\left(\frac{\tau_{k(s+t),n}/n-\tau_{k(s),n}/n}{t}\right)^{1/2}
\end{align*}
where $k(s)=\min \{i|\tilde X_{i:n,n}>s\}$.

The second term is the square root of
\begin{align}\label{stop_time_sum_eq}
\frac{\tau_{k(s+t)}-\tau_{k(s)}}{nt}
=\frac{1}{nt} \sum_{i=k(s)+1}^{k(s+t)} (\tau_{i,n}-\tau_{i-1,n}).
\end{align}
We have
\begin{align*}
E(\tau_{i,n}-\tau_{i-1:n}|\tilde X_1,\ldots,\tilde X_n)
=\sigma_{Y}^2(\tilde X_{i:n,n})
\end{align*}
so that the expectation of (\ref{stop_time_sum_eq}) given
$\tilde X_1,\ldots,\tilde X_n$ is
\begin{align}\label{st_given_x_eq}
\frac{1}{nt} \sum_{i=k(s)+1}^{k(s+t)} \sigma_{Y}^2(\tilde X_{i:n,n})
=\frac{1}{t} E_n \sigma_{Y}^2(\tilde X_i)I(s<X_i<s+t)
=\frac{E_n \sigma_{Y}^2(\tilde X_i)I(s<X_i<s+t)}
  {E\sigma_{Y}^2(\tilde X_i)I(s<X_i<s+t)}.
\end{align}
This converges to $1$ at a $\sqrt{n\underline \sigma^2_n/\log n}$ by
Theorem 37 in \citet{pollard_convergence_1984}.
We can bound the difference between (\ref{stop_time_sum_eq}) and
(\ref{st_given_x_eq}) using Bernstein's inequality conditional on the
$X_i$s.
We have, for any $k$ and $k'$,
\begin{align*}
&P\left(\frac{1}{nt} \left|\sum_{i=k}^{k'}
  (\tau_{i,n}-\tau_{i-1,n})
   -E(\tau_{i,n}-\tau_{i-1,n}|\tilde X_1,\ldots,\tilde X_n)\right|
  >\varepsilon\right)
\le 2\exp\left(-\frac{1}{2}
  \frac{n(t\varepsilon)^2}
     {K_1 +Mt\varepsilon/3}
  \right)
\end{align*}
where $K_1$ is a bound for the variance of the stopping times and $M$ is a
bound for their support.  For some constant $K_2$ and $t\varepsilon$
bounded from above, this is less than
\begin{align*}
2\exp(-n(t\varepsilon)^2/K_2)
\le 2\exp(-n(\underline\sigma_n\varepsilon)^2/K_2).
\end{align*}
Setting $\varepsilon=K_3 \sqrt{(\log n)/(n\underline \sigma^2_n)}$ for
$K_3$ large enough, this gives a bound of
$2\exp(-K_3^2\log n/K_2)$.
The probability of the supremum of the difference between
(\ref{stop_time_sum_eq}) and (\ref{st_given_x_eq}) being greater than
$\varepsilon$ is bounded by $n^2$ times this quantity ($(k(s),k(s+t))$ can
take on no more than $n^2$ values), and this can be made to go to zero
with $n$ by making $K_3$ large enough.

This argument shows that the following event holds with probability
approaching one for any sequence $a_n\to\infty$.  For any $(s,t)\in A_n$,
there are $(s',t')$ with $|s-s'|$ and $|t'-t|/t$ less than
$a_n\cdot (\log n)^{1/2}/(n^{1/2}\underline \sigma_n)$ such that
\begin{align*}
\frac{\sqrt{n}E_n\tilde Y_{i}I(s<\tilde X_i<s+t)}
  {\sqrt{t}}
=[\mathbb{B}_n(s+t)-\mathbb{B}_n(s)]/\sqrt{t}
   \cdot (1+\mathcal{O}_P((\log n)^{1/4}/(n^{1/4}\underline \sigma^{1/2}))).
\end{align*}
This completes the proof.

\end{proof}

We apply this lemma with 
\begin{align*}
A_n=\{(s,t)|(s,s+t)\in [0,\tilde x_{u}], t\ge \underline\sigma_n^2\}
\end{align*}
to get a lower bound of
\begin{align*}
\inf_{0\le s\le s+t\le \tilde x_{u}+b_n, t\ge \underline\sigma_n^2(1-b_n)}
[\mathbb{B}_n(s+t)-\mathbb{B}_n(s)]/\sqrt{t}
   \cdot (1+b_n^{1/2})
\end{align*}
with probability approaching one
for the term in (\ref{ks_stat_tilde_eq}) where
$b_n=a_n\cdot (\log n)^{1/2}/(n^{1/2}\underline \sigma_n)$.  By the
invariance properties of the Brownian motion, this has the same
distribution as
\begin{align*}
\inf_{0\le s\le s+t\le (\tilde x_{u}+b_n)/[\underline\sigma_n^2(1-b_n)], t\ge 1}
[\mathbb{B}(s+t)-\mathbb{B}(s)]/\sqrt{t}
   \cdot (1+b_n^{1/2})
\end{align*}
for a Brownian motion $\mathbb{B}$.
By arguments in \citet{kabluchko_extremes_2011},
we have, letting $c_n'=(\tilde x_{u}+b_n)/[\underline\sigma_n^2(1-b_n)]$
\begin{align*}
&P\left(
-\inf_{0\le s\le s+t\le (\tilde x_{u}+b_n)/[\underline\sigma_n^2(1-b_n)], t\ge 1}
[\mathbb{B}(s+t)-\mathbb{B}(s)]/\sqrt{t}
  \le (2\log c_n')^{1/2}
     + \frac{\frac{3}{2}\log\log c_n' - \log (2\pi^{1/2})+r}
         {(2\log c_n')^{1/2}}
\right)  \\
&\stackrel{n\to\infty}{\to} \exp(-\exp(-r)).
\end{align*}
This will hold with $c_n'$ replaced by
$c_n= \tilde x_{u}/\underline\sigma_n^2$ as long as
$c_n/c_n'\to 1$, which will hold as long as $b_n\to 0$.
This gives
\begin{align*}
&P\left(
\sqrt{n}\tilde T_{n}/(1+b_n^{1/2})
  \le (2\log c_n)^{1/2}
     + \frac{\frac{3}{2}\log\log c_n - \log (2\pi^{1/2})+r}
         {(2\log c_n)^{1/2}}
\right)  \\
&\stackrel{n\to\infty}{\to} \exp(-\exp(-r)).
\end{align*}
This will hold with $b_n$ replaced by $0$ if
\begin{align*}
b_n^{1/2}\log c_n
=a_n^{1/2}\cdot (\log n)^{1/4}/(n^{1/4}\underline \sigma_n^{1/2})
  \log (\tilde x_{u}/\underline\sigma_n)
\to 0
\end{align*}
which will hold for $a_n$ increasing slowly enough as long as
\begin{align}\label{sigma_rate_eq}
\underline \sigma_n n^{1/2}/(\log n)^{5/2}\to\infty.
\end{align}

To get the same extreme value limit with $\sigma(s,t)$ replaced by
$\hat\sigma(s,t)$, it suffices to show that
\begin{align}\label{sigma_diff_eq}
\frac{\hat\sigma(s,t)-\tilde\sigma(s,t)}{\tilde\sigma(s,t)}
\end{align}
converges to $1$ at a $(\log n)$ rate uniformly in $(s,t)$ such that
$\tilde\sigma(s,t)\ge \underline \sigma_n^2/2$.
To show
that replacing $\sigma(s,t)$ with its estimate gives a limit that is
asymptotically no greater than an extreme value random variable, it
suffices to show that this quantity is bounded from below by
$1-o_p(\log n)$ term, and converges to $1$ at any rate (the latter
condition is needed for the sets on which the infimum is taken to converge
quickly enough).  We have
\begin{align*}
&E\hat\sigma^2(s,t)=var(Y_iI(s<X<s+t))  \\
&=E[var(Y_i|X_i)I(s<X<s+t)]
+var(E(Y_i|X_i)I(s<X<s+t))  \\
&=\tilde\sigma^2(s,t)+var(E(Y_i|X_i)I(s<X<s+t)).
\end{align*}
By Theorem 37 in \citet{pollard_convergence_1984},
$(\hat\sigma^2(s,t)-E\hat\sigma^2(s,t))/\tilde\sigma^2(s,t)$
converges to $1$ at a $(n\underline\sigma^2_n/\log n)^{1/2}$ rate
uniformly in $\tilde\sigma(s,t)\ge \underline\sigma_n/2$, which translates
into at least a
$(n\underline\sigma^2_n/\log n)^{1/4}$ rate for
\begin{align*}
\frac{\hat\sigma(s,t)-\sqrt{E\hat\sigma^2(s,t))}}{\tilde\sigma(s,t)}
\ge \frac{\hat\sigma(s,t)-\tilde\sigma(s,t)}{\tilde\sigma(s,t)},
\end{align*}
which will be fast enough as long as (\ref{sigma_rate_eq}) holds.  For
(\ref{sigma_diff_eq}) to converge to $1$, it now suffices to have
\begin{align*}
\frac{var(E(Y_i|X_i)I(s<X<s+t))}{\tilde\sigma^2(s,t)}
\le C var(E(Y_i|X_i)I(s<X<s+t))/t\to 0.
\end{align*}
Consistency of $\hat\sigma(s,t)$ also implies that $\hat c_n/c_n$
converges in probability to one, which allows $c_n$ to be replaced with
$\hat c_n$.

\bibliography{library}

\end{document}